\documentclass{article}%
\usepackage{amsmath}
\usepackage{amsfonts}
\usepackage{amssymb}
\usepackage{graphicx}%
\providecommand{\U}[1]{\protect\rule{.1in}{.1in}}
\newtheorem{theorem}{Theorem}

\newtheorem{lemma}[theorem]{Lemma}

\newenvironment{proof}[1][Proof]{\noindent\textbf{#1.} }{\ \rule{0.5em}{0.5em}}
\begin{document}

\title{An example of a quantum statistical model which cannot be mapped to a less
informative one by any trace preserving positive map}
\author{Keiji Matsumoto\\National Institute of Informatics, \\2-1-2 Hitotsubashi, Chiyoda-ku, Tokyo 101-8430 \\e-mail:keiji@nii.ac.jp}
\maketitle

Comparison of statistical models (experiments) is an important branch of
mathematical statistics, which gives deep insights in many aspects of
foudation of statistics (see \cite{ShiryaevSpokoiny}\cite{Torgersen} and
references therein, for example). Given a parameter space $\Theta$ (for
simplicity, here we suppose $\Theta$ is a finte set), consider two families of
probability distirbutions $\left\{  p_{\theta}\right\}  _{\theta\in\Theta}$
and $\left\{  q_{\theta}\right\}  _{\theta\in\Theta}$ . \ Let $\mathcal{D}$ be
a space of the decision that a statistician may take, and for each
$d\in\mathcal{D}$, let $l_{\theta}\left(  d\right)  $ be the loss when
decision $d$ is chosen and true parameter value is $\theta$. We say the former
is more informative than the latter, or
\[
\left\{  p_{\theta}\right\}  _{\theta\in\Theta}\geq\left\{  q_{\theta
}\right\}  _{\theta\in\Theta}\,,
\]
if and only if the following is true for all $\mathcal{D}$ and $\left\{
l_{\theta}\right\}  _{\theta\in\Theta}$ : for any the map $d$ from the data
space to the decision space, there is $d^{\prime}$ such that
\[
\sum_{x}l_{\theta}\left(  d^{\prime}\left(  x\right)  \right)  p_{\theta
}\left(  x\right)  \leq\sum_{x}l_{\theta}\left(  d\left(  x\right)  \right)
q_{\theta}\left(  x\right)  ,\,\forall\theta\in\Theta.
\]
Celebrated Brackwell theorem, or randoization criteria of LeCam
\cite{ShiryaevSpokoiny}\cite{Torgersen} tells that this is equivalent to the
existence of the transition probability $P\left(  x|x^{\prime}\right)  $
\[
p_{\theta}\left(  x\right)  =\sum_{x^{\prime}}P\left(  x|x^{\prime}\right)
q_{\theta}\left(  x^{\prime}\right)  ,\,\forall\theta\in\Theta.
\]
(For notational simplicity, here and below the set of all data ($x$'s) and the
set where the descison takes value has finite number of elements. However,
with proper mathematical settings, the above results essentially holds even if
the former and the latter is an arbitrary measurable space and an arbitrary
topological space, respectively.)  

So far, there are two versions of quantum extensions of this concept
\cite{Buscemi:10}\cite{Jencova:12}\cite{Jencova:14}\cite{Matsumoto:10}. One is
to consider classical decision problem on quantum state families.
\[
\left\{  \rho_{\theta}\right\}  _{\theta\in\Theta}\geq^{c}\left\{
\sigma_{\theta}\right\}  _{\theta\in\Theta}%
\]
if and only if the following is true for all $\mathcal{D}$ and $\left\{
l_{\theta}\right\}  _{\theta\in\Theta}$ : for any measurement $\left\{
M_{d}\right\}  _{d\in\mathcal{D}}$ taking values in $\mathcal{D}$, there is a
mesurement $\left\{  M_{d}^{\prime}\right\}  _{d\in\mathcal{D}}$ such that
\[
\sum_{d\in\mathcal{D}}l_{\theta}\left(  d\right)  \mathrm{tr}\,M_{d}^{\prime
}\,\rho_{\theta}\leq\sum_{d\in\mathcal{D}}l_{\theta}\left(  d\right)
\mathrm{tr}\,M_{d}\,\sigma_{\theta},,\,\forall\theta\in\Theta.
\]

Another version is to consider the full quantum task. Consider a Hilbert space
$\mathcal{H}_{D}$ as an quantum analogue of decision space, and an operator
$L_{\theta}$ in $\mathcal{H}_{D}$ defining the loss. A decision rule is a
completely positive trace preserving (CPTP) map $\Lambda$ to operators on
$\mathcal{H}_{D}$ . Then we write
\[
\left\{  \rho_{\theta}\right\}  _{\theta\in\Theta}\geq^{q}\left\{
\sigma_{\theta}\right\}  _{\theta\in\Theta}\,
\]
if and only if the following is true for all $\mathcal{H}_{D}$ and $\left\{
L_{\theta}\right\}  _{\theta\in\Theta}$: for any CPTP map $\Lambda$ there is a
CPTP map $\Lambda^{\prime}$ such that
\[
\mathrm{tr}\,L_{\theta}\Lambda^{\prime}\left(  \rho_{\theta}\right)
\leq\mathrm{tr}\,L_{\theta}\Lambda\left(  \sigma_{\theta}\right)
,\,\forall\theta\in\Theta.
\]
(Here, the loss  measure is linear in the state. But, use of  bounded and
continuous functionals of the state over $\mathcal{H}_{D}$ does not change the
definition of $\geq^{q}$ at all \cite{Matsumoto:10}.)   

"$\geq^{q}$" holds if and only if there is a CPTP map $\Gamma$ such
that\thinspace\cite{Jencova:14}\cite{Matsumoto:10}
\begin{equation}
\Gamma\left(  \rho_{\theta}\right)  =\sigma_{\theta},\,\theta\in
\Theta.\label{transform}%
\end{equation}
Meantime, if there is positive trace preserving $\Gamma$ with above relation
exists, then "$\geq^{c}$" holds obviously. A natural question is whether this
is necessary. In this paper, we answer the question negatively by giving a
couter example.

Let $\Theta=\left\{  0,1\right\}  $, and consider the following condition,
\begin{equation}
\left\Vert \rho_{0}-t\rho_{1}\right\Vert _{1}\geq\left\Vert \sigma
_{0}-t\,\sigma_{1}\right\Vert _{1},\,\forall t\geq0,\label{r-tr-e}%
\end{equation}
where $\left\Vert A\right\Vert _{1}:=\mathrm{tr}\,\sqrt{A^{\dagger}A}$. If
$g\left(  x\right)  $ is a real valued function, $\left\Vert g\right\Vert
_{1}:=\sum_{x}\left\vert g\left(  x\right)  \right\vert $.

We use the following lemma. This is not new \cite{Jencova:12}%
\cite{Matsumoto:10}, but the proof is stated for completeness.  

\begin{lemma}
\cite{Jencova:12}\cite{Matsumoto:10} Suppose $\left[  \rho_{0},\rho
_{1}\right]  =0$. Then, $\left\{  \rho_{\theta}\right\}  _{\theta\in\Theta
}\mathcal{\geq}^{c}\left\{  \sigma_{\theta}\right\}  _{\theta\in\Theta}$ if
and only if (\ref{r-tr-e}) holds.
\end{lemma}

\begin{proof}
Let $q_{\theta}^{M}\left(  x\right)  :=\mathrm{tr}\,\sigma_{\theta}M_{x}$,
where $M=\left\{  M_{x}\right\}  $ is a POVM, $M_{x}\geq0$, $\sum_{x}%
M_{x}=\mathbf{1}$. Then, by definition, $\left\{  \rho_{\theta}\right\}
_{\theta\in\Theta}\mathcal{\geq}^{c}\left\{  \sigma_{\theta}\right\}
_{\theta\in\Theta}$ \ if and only if
\[
\,\,\left\{  \rho_{\theta}\right\}  _{\theta\in\Theta}\geq\left\{  q_{\theta
}^{M}\right\}  _{\theta\in\Theta},\forall M.\,
\]
This is equivalent to \cite{Torgersen-finite}
\[
\,\left\Vert \rho_{0}-t\rho_{1}\right\Vert _{1}\geq\left\Vert q_{0}^{M}%
-tq_{1}^{M}\right\Vert _{1},\,\forall M\,,\,\forall t\geq0.\,
\]
Therefore, since
\[
\,\max_{M}\left\Vert q_{0}^{M}-tq_{1}^{M}\right\Vert _{1}=\left\Vert
\sigma_{0}-t\sigma_{1}\right\Vert _{1},
\]
we have the assertion.
\end{proof}

Below, we give an example such that $\left\{  \rho_{\theta}\right\}
_{\theta\in\Theta}\mathcal{\geq}^{c}\left\{  \sigma_{\theta}\right\}
_{\theta\in\Theta}$ but there is no trace preserving positive map $\Gamma$
with (\ref{transform}). The example is given as follows. Let%

\begin{align*}
\rho_{0} &  =\left[
\begin{array}
[c]{ccc}%
\alpha &  & \\
& 0 & \\
&  & 1-\alpha
\end{array}
\right]  ,\rho_{1}=\left[
\begin{array}
[c]{ccc}%
0 &  & \\
& \alpha & \\
&  & 1-\alpha
\end{array}
\right]  ,\\
1 &  \geq\,\alpha\geq0,\,
\end{align*}
and%

\[
\sigma_{0}=\left[
\begin{array}
[c]{cc}%
1 & 0\\
0 & 0
\end{array}
\right]  ,\sigma_{1}=\left[
\begin{array}
[c]{cc}%
1-\beta^{2} & \sqrt{1-\beta^{2}}\beta\\
\sqrt{1-\beta^{2}}\beta & \beta^{2}%
\end{array}
\right]  .
\]
Further, we suppose
\[
\alpha\geq\beta.
\]

It is very easy to see that there is no trace preserving positive map $\Gamma$
with (\ref{transform}). The proof runs as follows. Suppose there is such a
positive map. Then for any pure state $\psi_{0}$ and $\psi_{1}$ in the support
of $\rho_{0}$ and $\rho_{1}$ respectively, we should have
\[
\Gamma\left(  \left\vert \psi_{\theta}\right\rangle \left\langle \psi_{\theta
}\right\vert \right)  =\sigma_{\theta},\,\theta\in\Theta.
\]
However, $\left(  0\,0\,1\right)  ^{T}$ is a common element of the support of
$\rho_{0}$ and that of $\rho_{1}$.  Hence, with $\psi_{0}=\psi_{1}=\left(
0\,0\,1\right)  ^{T}$, the above equation is impossible.

In addition, as will be shown in the following by elementary analysis,
ifholds, (\ref{r-tr-e}) is true. Hence, by the above lemma, this means
$\left\{  \rho_{\theta}\right\}  _{\theta\in\Theta}\mathcal{\geq}^{c}\left\{
\sigma_{\theta}\right\}  _{\theta\in\Theta}$ holds. Therefore, this is an
example we need.  

The proof is as follows. If $t<0$, (\ref{r-tr-e}) holds obviously. Thus, with
\[
f\left(  t\right)  :=\left\Vert \rho_{0}-t\rho_{1}\right\Vert _{1}%
^{2}-\left\Vert \sigma_{0}-t\,\sigma_{1}\right\Vert _{1}^{2}\,,
\]
we prove $f\left(  t\right)  \geq0$ for all $t\geq0$.

If $t\geq0$,%

\begin{align*}
\left\Vert \rho_{0}-t\rho_{1}\right\Vert _{1} &  =\alpha\left(  1+t\right)
+\left(  1-\alpha\right)  \left\vert 1-t\right\vert ,\\
\left\Vert \sigma_{0}-t\sigma_{1}\right\Vert _{1} &  =\sqrt{\left(
1-t\right)  ^{2}+4t\beta^{2}}\leq\sqrt{\left(  1-t\right)  ^{2}+4t\alpha^{2}}%
\end{align*}
In the case of $0\leq t\leq1$,
\[
f\left(  t\right)  \geq4\alpha\left(  1-\alpha\right)  \left(  -t^{2}%
+t\right)  \geq0.
\]
In the case of $t\geq1$,
\[
f\left(  t\right)  \geq4\alpha\left(  1-\alpha\right)  \left(  t-1\right)
\geq0.
\]
After all, we have $f\left(  t\right)  \geq0$ for all $\,t\geq0$.

\end{document}